\documentclass[a4paper]{article}
\usepackage{\jobname}

\title{Generalized Complexity of $\mathcal{ALC}$ Subsumption}

\author{Arne Meier\\Institut f\"ur Theoretische Informatik, Leibniz Universit\"at Hannover\\meier@thi.uni-hannover.de}
\date{}

\begin{document}
\maketitle
\thispagestyle{empty}
\begin{abstract}
The subsumption problem with respect to terminologies in the description logic \ALC is $\EXPTIME$-complete. We investigate the computational complexity of fragments of this problem by means of allowed Boolean operators. Hereto we make use of the notion of clones in the context of Post's lattice. Furthermore we consider all four possible quantifier combinations for each fragment parameterized by a clone. We will see that depending on what quantifiers are available the classification will be either tripartite or a quartering.
\end{abstract}

\section{Introduction}
Description logics (DL) play an important role in several areas of research, e.g., semantic web, database structuring, or medical ontologies \cite{DLHB,brle84,brsc85,owl2}. As a consequence there exists a vast range of different extensions where each of them is highly specialized to its field of application. 
Nardi and Brachman describe \emph{subsumption} as the important inference problem within DL \cite{nabra03}. Given two (w.l.o.g.\ atomic) concepts $C,D$ and a set of axioms (which are pairs of concept expressions $A,B$ stating $A$ implies $B$), one asks the question whether $C$ implies $D$ is consistent with respect to each model satisfying all of the given axioms. Although the computational complexity of subsumption in general can be  between tractable and $\EXPTIME$ \cite{dlnhnm92,dlnn97} depending on which feature\footnote{These features, for instance, can be existential or universal restrictions, availability of disjunction, conjunction, or negation. Furthermore, one may extend a description logic with more powerful concepts, e.g., number restrictions, role chains, or epistemic operators \cite{DLHB}.} is available, there exist nonetheless many DL which provide a tractable (i.e., in $\P$) subsumption reasoning problem, e.g., the DL-Lite and $\mathcal{EL}$ families \cite{ackz07,ackz09,bbl05,bbl08,cgllr05}. One very prominent application example of the subsumption problem is the SNOMED CT clinical database which includes about $400.000$ axioms and is a subset of the DL $\mathcal{EL}^{++}$ \cite{19007439,19007443}.

In this paper we investigate the subsumption problem with respect to the most general (in sense of available Boolean operators)  description logic $\ALC$. It is known that the unrestricted version of this problem is $\EXPTIME$-complete due to reducibility to a specific DL satisfiability problem \cite{DLHB}, and is therefore highly intractable. Our aim is to understand where this intractability comes from or to which Boolean operator it may be connected to. Therefore we will make use of the well understood and much used algebraic tool, Post's lattice \cite{pos41}. At this approach one describes sets formulae (w.r.t.\ to some \emph{clone}) which can be constructed from a given finite set of Boolean functions by arbitrary composition of these functions, or projections. For a good introduction into this area consider \cite{Bohler:2003tg}. The main technique is to investigate fragments of a specific decision problem by means of allowed Boolean functions; in this paper this will be the subsumption problem. As Post's lattice considers any possible set of all Boolean functions a classification by it always yields an exhaustive study. This kind of research  has been done previously for several different kind of logics, e.g., temporal, hybrid, modal, and nonmonotonic logics \cite{bemethvo09,cmtv10,crscthwo10,hescsc08,sc07,sc08}.

\paragraph{Main results.} The most general class of fragments, i.e., those which have both quantifiers available perfectly show how powerful the subsumption problem is. Having access to at least one constant (true or false) leads to an intractable fragment. Merely for the fragment where only projections (and none of the constants) are present it is not clear if there can be a polynomial time algorithm for this case and has been left open. If one considers the cases where only one quantifier is present, then the fragments around disjunction (case $\forall$), respectively, the ones around conjunction (case $\exists$) become tractable. Without quantifiers conjunctive and disjunctive fragments are $\P$-complete whereas the fragments which include either the affine functions (exclusive or), or can express $x\lor(y\land z)$, or $x\land(y\lor z)$, or self-dual functions (i.e., $f(x_1,\dots,x_n)=\lnot f(\lnot x_1,\dots,\lnot x_n)$) are intractable. 
\Cref{fig:SUBSexall_ex_all} depicts how the results directly arrange within Post's lattice. 


\section{Preliminaries}
In this paper we will make use of standard notions of complexity theory \cite{pip97b}. In particular, we work with the classes $\NL, \P, \co\NP, \EXPTIME$, and the class $\oplus\LOGSPACE$ which corresponds to those nondeterministic Turing machine running in logarithmic space whose computations trees have an odd number of accepting paths. Usually all stated reductions are logarithmic space many-one reductions $\leqlogm$. We write $A\equivlogm B$ iff $A\leqlogm B$ and $B\leqlogm A$ hold.

\paragraph{Post's Lattice.} Let $\true$, $\false$ denote the truth values true, false. Given a finite set of Boolean functions $B$, we say the \emph{clone of $B$} contains all compositions of functions in $B$ plus all projections; the smallest such clone is denoted with $[B]$ and the set $B$ is called a \emph{base of $[B]$}. The lattice of all clones has been established in \cite{pos41} and a much more succinct introduction can be found in \cite{Bohler:2003tg}. \Cref{tab:base} depicts all clones and their bases which are relevant for this paper. Here $\mathrm{maj}$ denotes the majority, and $\mathrm{id}$ denotes identity. Let $f\colon\{\true,\false\}^n\to\{\true,\false\}$ be a Boolean function. Then the \emph{dual of $f$}, in symbols $\dual f$, is the $n$-ary function $g$ with $g(x_1,\dots,x_n)=\overline{f(\overline{x_1},\dots,\overline{x_n})}$. Similarly, if $B$ is a set of Boolean functions, then $\dual B:=\{\dual f\mid f\in B\}$. Further, abusing notation, define $\dual\exists:=\forall$ and $\dual\forall=\exists$; if $\calQ\subseteq\{\exists,\forall\}$ then $\dual\calQ:=\{\dual\Game\mid\Game\in\calQ\}$.

\begin{table}
$$
\begin{array}{cc|cc|cc}
 \text{Clone} & \text{Base} &  \text{Clone} & \text{Base} & \text{Clone} & \text{Base} \\\hline
 \CloneBF & \{x\land y,\overline x\} &
 \CloneS_{00} & \{x\lor(y\land z)\} &
 \CloneS_{10} & \{x\land(y\lor z)\} \\
 \CloneD_1 & \{\mathrm{maj}\{x,y,\overline{z}\}\} &
 \CloneD_2 & \{\mathrm{maj}\{x,y,z\}\} &
 \CloneM_0 & \{x\land y, x\lor y, \false\}\\
 \CloneL & \{x\xor y,\true\} &
 \CloneL_0 & \{x\xor y\} &
 \CloneL_1 & \{x\leftrightarrow y\} \\
 \CloneL_2 & \{x\xor y\xor z\} &
 \CloneL_3 & \{x\xor y\xor z\xor\true\} \\
 \CloneV & \{x\lor y,\true,\false\} &
 \CloneV_0 & \{x\lor y,\false\} &
 \CloneV_2 & \{x\lor y\} \\
 \CloneE & \{x\land y,\true,\false\}&
 \CloneE_0 & \{x\land y,\false\} &
 \CloneE_2 & \{x\land y\} \\
 \CloneN & \{\overline x,\true\} &
 \CloneN_2 & \{\overline x\} &\\
 \CloneI_0 & \{\mathrm{id},\false\}&
 \CloneI_1 & \{\mathrm{id},\true\}&
 \CloneI_2 & \{\mathrm{id}\}
\end{array}
$$
\caption{All clones and bases relevant for this paper.}\label{tab:base}
\end{table}

\paragraph{Description Logic.} We use the standard syntax and semantics of $\mathcal{ALC}$ as in \cite{DLHB}. Additionally we we adjusted them to fit the notion of clones. 
The set of \emph{concept descriptions} (or \emph{concepts}) is defined by
$
C := A\mid \circ_f(C,\dots,C)\mid \exists R.C\mid \forall R.C,
$ 
where $A$ is an atomic concept (variable), $R$ is a role (transition relation), and $\circ_f$ is a Boolean operator which corresponds to a Boolean function $f\colon\{\true,\false\}^n\to\{\true,\false\}$. For a given set $B$ of Boolean operators and $\calQ\subseteq\{\exists,\forall\}$, we define that a \emph{$B$-$\calQ$-concept} uses only operators from $B$ and quantifiers from $\calQ$. Hence, if $B=\{\land,\lor\}$ then $[B]=\CloneBF$, and the set of $B$-concept description is equivalent to (full) $\mathcal{ALC}$. Otherwise if $[B]\subsetneq\CloneBF$ for some set $B$, then we consider real subsets of $\mathcal{ALC}$ and cannot express any (usually in $\mathcal{ALC}$ available) concept. 
An \emph{axiom} is of the form $C\dsub D$, where $C$ and $D$ are concepts; $C\equiv D$ is the syntacic sugar for $C\dsub D$ and $D\dsub C$. A \emph{TBox} is a finite set of axioms and a $B$-$\calQ$-TBox contains only axioms of $B$-$\calQ$-concepts.

An \emph{interpretation} is a pair $\calI=(\Delta^\calI,\cdot^\calI)$, where $\Delta^\calI$ is a nonempty set and $\cdot^\calI$ is a mapping from the set of atomic concepts to the power set of $\Delta^\calI$, and from the set of roles to the power set of $\Delta^\calI\times\Delta^\calI$. We extend this mapping to arbitrary concepts as follows:
\begin{align*}
 (\exists R.C)^\calI &= \big\{x\in\Delta^\calI\,\big|\, \{y\in C^\calI\mid (x,y)\in R^\calI\}\neq\emptyset\big\},\\
 (\forall R.C)^\calI &= \big\{x\in\Delta^\calI\,\big|\, \{y\in C^\calI\mid (x,y)\notin R^\calI\}=\emptyset\big\},\\
 \big(\circ_f(C_1,\dots,C_n)\big)^\calI &= \big\{x\in\Delta^\calI\,\big|\, f(||x\in C_1^\calI||,\dots,||x\in C^\calI_n||)=\true\big\},
\end{align*}
where $||x\in C^\calI||= 1$ if $x\in C^\calI$ and $||x\in C^\calI||=0$ if $x\notin C^\calI$. An interpretation $\calI$ \emph{satisfies} the axiom $C\dsub D$, in symbols $\calI\models C\dsub D$, if $C^\calI\subseteq D^\calI$. Further $\calI$ \emph{satisfies a TBox}, in symbols $\calI\models\calT$, if it satisfies every axiom therein; then $\calI$ is called a \emph{model}. Let $\calQ\subseteq\{\exists,\forall\}$ and $B$ be a finite set of Boolean operators. Then for the \emph{TBox-concept satisfiability problem, $\TCSAT_\calQ(B)$,} given a $B$-$\calQ$-TBox $\calT$ and a $B$-$\calQ$-concept $C$, one asks if there is an $\calI$ s.t.\ $\calI\models\calT$ and $C^\calI\neq\emptyset$. This problem has been fully classified w.r.t.\ Post's lattice in \cite{ms11}. Further the \emph{Subsumption problem, $\SUBS_\calQ(B)$,} given a $B$-$\calQ$-TBox and two $B$-$\calQ$-concepts $C,D$, asks if for every interpretation $\calI$ it holds that $\calI\models\calT$ implies $C^\calI\subseteq D^\calI$.\bigskip

As subsumption is an inference problem within DL some kind of connection in terms of reductions to propositional implication is not devious. 
In \cite{bmtv09} Beyersdorff \etal\ classify the propositional implication problem $\IMP$ with respect to all fragments parameterized by all Boolean clones. 

\begin{theorem}[\cite{bmtv09}]\label{thm:imp}
Let $B$ be a finite set of Boolean operators.
\begin{enumerate}
	\item If $C\subseteq[B]$ for $C\in\{\CloneS_{00},\CloneD_2,\CloneS_{10}\}$, then $\IMP(B)$ is $\co\NP$-complete \wrt $\leq^{\AC0}_m$\footnote{A language $A$ is $\AC0$ many-one reducible to a language $B$ ($A\leq^{\AC0}_m B$) if there exists a function $f$ computable by a logtime-uniform $\AC0$-circuit familiy such that $x\in A$ iff $f(x)\in B$ (for more information, see \cite{vol99}).}.
	\item If $\CloneL_2\subseteq[B]\subseteq\CloneL$, then $\IMP(B)$ is $\xor\LOGSPACE$-complete \wrt $\leq^{\AC0}_m$.
	\item If $\CloneN_2\subseteq[B]\subseteq\CloneN$, then $\IMP(B)$ is in $\AC0[2]$.
	\item Otherwise $\IMP(B)\in\AC0$.
\end{enumerate} \NoEndMark
\end{theorem}
\section{Interreducibilities}
The next lemma proves base independence for the subsumption problem. This is important to generalize the results to clones corresponding to their respective base. In particular, this kind of property enables us to use standard bases for every clone within our proofs. The result is proven in the same way as in \cite[Lemma 4]{ms11b}.

\begin{lemma}\label{lem:baseind}\NoEndMark
 Let $B_1,B_2$ be two sets of Boolean operators such that $[B_1]\subseteq[B_2]$, and let $\calQ \subseteq \{\exists,\forall\}$.
Then $\SUBS_\calQ(B_1)\leqlogm\SUBS_\calQ(B_2)$.
\end{lemma}

The following two lemmata deal with a duality principle of subsumption. The correctness of contraposition for axioms allows us to state reduction to the fragment parameterized by the dual operators. Further having access to negation allows us in the same way as in \cite{MS10} to simulate both constants.

\begin{lemma}\label[lemma]{lem:contrapositionSUBS} \NoEndMark
	Let $B$ be a finite set of Boolean operators and $\calQ\subseteq\{\forall,\exists\}$. 	
Then $\SUBS_\calQ(B)\leqlogm\SUBS_{\dual\calQ}(\dual B)$.
\end{lemma}
\begin{proof}
 Here we distinguish two cases. Given a concept $A$ define with $A^\lnot$ the concept $\lnot A$ in negation normal form (NNF). 
 
 First assume that $\lnot\in[B]$. Then $(\calT,C,D)\in\SUBS_\calQ(B)$ if and only if for any interpretation $\calI$ s.t.\ $\calI\models\calT$ it holds that $C^\calI\subseteq D^\calI$ if and only if for any interpretation $\calI$ s.t.\ $\calI\models\calT':=\{F^\lnot\dsub E^\lnot\mid E\dsub F\in\calT\}$ it holds that $(\lnot D)^\calI\subseteq(\lnot C)^\calI$ if and only if $(\calT',D^\lnot,C^\lnot)\in\SUBS_{\dual\calQ}(\dual B)$. The correctness directly follows from $\dual\lnot=\lnot$.
 
 Now assume that $\lnot\notin[B]$. Then for a given instance $(\calT,C,D)$ it holds that for the contraposition instance $(\{F^\lnot\dsub E^\lnot\mid E\dsub F\in\calT\}, D^\lnot, C^\lnot)$ before every atomic concept occurs a negation symbol. Denote with $(\{F^\lnot\dsub E^\lnot\mid E\dsub F\in\calT\}, D^\lnot, C^\lnot)^{\textrm{pos}}$ the substitution of any such  negated atomic concept $\lnot A$ by a fresh concept name $A'$. Then $(\calT,C,D)\in\SUBS_\calQ(B)$ iff $(\{F^\lnot\dsub E^\lnot\mid E\dsub F\in\calT\}, D^\lnot, C^\lnot)^{\textrm{pos}}\in\SUBS_{\dual\calQ}(\dual B)$.
\end{proof}

\begin{lemma}\label[lemma]{lem:topbot-always-above-negSUBS} \NoEndMark
  Let $B$ be a finite set of Boolean operators s.t. $\CloneN_2\subseteq[B]$ and $\calQ \subseteq \{\exists,\forall\}$. 
 Then it holds that $\SUBS_\calQ(B)\equivlogm\SUBS_\calQ(B\cup\{\true,\false\})$.
\end{lemma}

Using Lemma 4.2 in \cite{bmtv09} we can easily obtain the ability to express the constant $\true$ whenever we have access to conjunctions, and the constant $\false$ whenever we are able to use disjunctions.

\begin{lemma}\label[lemma]{lem:constants-for-s00-s01} \NoEndMark
Let $B$ be a finite set of Boolean operators and $\calQ\subseteq\{\forall,\exists\}$.
\begin{enumerate}
	\item If $\CloneE_0\subseteq[B]$, then $\SUBS_\calQ(B)\equivlogm\SUBS_\calQ(B\cup\{\true\})$.
	\item If $\CloneV_0\subseteq[B]$, then $\SUBS_\calQ(B)\equivlogm\SUBS_\calQ(B\cup\{\false\})$.
\end{enumerate}
\end{lemma}

The connection of subsumption to terminology satisfiability and propositional implication is crucial for stating upper and lower bound results. The next lemma connects subsumption to $\TCSAT$ and also to $\IMP$.

\begin{lemma}\label[lemma]{lem:subs-csat} \NoEndMark
	Let $B$ be a finite set of Boolean operators and $\calQ\subseteq\{\forall,\exists\}$ be a set of quantifiers. Then the following reductions hold:
	\begin{enumerate}
		\item $\IMP(B)\leqlogm\SUBS_\emptyset(B)$.\label{lemitem:imp_lowerbound}
		\item $\SUBS_\calQ(B)\leqlogm\overline{\TCSAT_\calQ(B\cup\{\nimp\})}$.\label{lemitem:coTCSAT_upperbound}
		\item $\overline{\TCSAT_\calQ(B)}\leqlogm\SUBS_\calQ(B\cup\{\false\})$.\label{lemitem:coTCSAT_lowerbound}
	\end{enumerate}
\end{lemma}

\begin{proof}
\begin{enumerate}
	\item Holds due to $(\varphi,\psi)\in\IMP(B)$ iff $(C_\varphi,C_\psi,\emptyset)\in\SUBS_\emptyset(B)$, for concept descriptions $C_\varphi=f(\varphi),C_\psi=f(\psi)$ with $f$ mapping propositional formulae to concept descriptions via
\begin{align*}
	f(\true) &= \true, \text{ and } f(\false) = \false,\\
	f(x) &= C_x, \text{ for variable }x,\\
	f(g(C_1,\dots,C_n)) &= \fop g(f(C_1),\dots,f(C_n))
\end{align*}
where $g$ is an $n$-ary Boolean function and $\fop g$ is the corresponding operator.
	\item $(C,D,\calT)\in\SUBS_\calQ(B)$ iff $(\calT,C\nimp D)\in\overline{\TCSAT_\calQ(B\cup\{\nimp\})}$. \cite{DLHB}.
	\item $(\calT,C)\in\overline{\TCSAT_\calQ(B)}$ iff $(C,\false,\calT)\in\SUBS_\calQ(B\cup\{\false\})$. \cite{DLHB}.
\end{enumerate}

\end{proof}
\section{Main Results}
We will start with the subsumption problem using no quantifiers and will show that the problem either is $\co\NP$-, $\P$-, $\NL$-complete, or is $\xor\LOGSPACE$-hard.

\begin{theorem}[No quantifiers available.]\label{thm:subs_empty} \NoEndMark
Let $B$ be a finite set of Boolean operators.
\begin{enumerate}
	\item If $X\subseteq[B]$ for $X\in\{\CloneL_0,\CloneL_1,\CloneL_3,\CloneS_{10},\CloneS_{00},\CloneD_2\}$, then $\SUBS_\emptyset(B)$ is $\co\NP$-complete.
	\item If $\CloneE_2\subseteq[B]\subseteq\CloneE$ or $\CloneV_2\subseteq[B]\subseteq\CloneV$, then $\SUBS_\emptyset(B)$ is $\P$-complete.
	\item If $[B]=\CloneL_2$, then $\SUBS_\emptyset(B)$ is $\xor\LOGSPACE$-hard.
	\item If $\CloneI_2\subseteq[B]\subseteq\CloneN$, then $\SUBS_\emptyset(B)$ is $\NL$-complete.
\end{enumerate}
All hardness results hold \wrt $\leqlogm$ reductions.
\end{theorem}
\begin{proof}
	\begin{enumerate}
		\item The reduction from the implication problem $\IMP(B)$ in \Cref{lem:subs-csat}(\ref{lemitem:imp_lowerbound}.) in combination with \Cref{thm:imp} and \Cref{lem:baseind} proves the $\co\NP$ lower bounds of $\CloneS_{10},\CloneS_{00}, \CloneD_2$.
		The lower bounds for $\CloneL_0\subseteq[B]$ and $\CloneL_3\subseteq[B]$ follow from \Cref{lem:subs-csat}(\ref{lemitem:coTCSAT_lowerbound}.) with $\overline{\TCSAT_\emptyset(B)}$ being $\co\NP$-complete which follows from the $\NP$-completeness result of $\TCSAT_\emptyset(B)$ shown in \cite[Theorem 27]{ms11b}. 
		Further the lower bound for $\CloneL_1\subseteq[B]$ follows from the duality of '$\xor$' and '$\equiv$' and \Cref{lem:contrapositionSUBS} with respect to the case $\CloneL_0\subseteq[B]$ enables us to state the reduction
		$$
		\SUBS_\emptyset(\CloneL_0)\leqlogm\SUBS_{\dual\emptyset}(\dual{\CloneL_0})=\SUBS_\emptyset(\CloneL_1).
		$$
		 
		 The upper bound follows from a reduction to $\overline{\TCSAT_\emptyset(\CloneBF)}$ by \Cref{lem:subs-csat}(\ref{lemitem:coTCSAT_upperbound}.) and the membership of $\TCSAT_\emptyset(\CloneBF)$ in $\NP$ by \cite[Theorem 27]{ms11b}.
		 
		\item The upper bound follows from the memberships in $\P$ for $\SUBS_\exists(\CloneE)$ and $\SUBS_\forall(\CloneV)$ proven in \Cref{thm:subs_ex,thm:subs_all}.
		
		The lower bound for $[B]=\CloneE_2$ follows from a reduction from the hypergraph accessibility problem\footnote{
		In a given hypergraph $H=(V,E)$, a hyperedge $e\in E$ is a pair of source nodes $\textit{src}(e)\in V\times V$ and one destination node $\textit{dest}(e)\in V$. Instances of $\HGAP$ consist of a directed hypergraph $H=(V,E)$, a set $S\subseteq V$ of source nodes, and a target node $t\in V$. Now the question is whether there exists a hyperpath from the set $S$ to the node $t$, i.e., whether there are hyperedges $e_1,e_2,\dots,e_k$ such that, for each $e_i$, there are $e_{i_1},\dots,e_{i_\nu}$ with $1\leq i_1,\dots,i_\nu<i$ and $\bigcup_{j\in\{i_1,\dots,i_\nu\}}\textit{dest}(e_j)\cup \textit{src}(e_j)\supseteq \textit{src}(e_i)$, and $\textit{src}(e_1)=S$ and $\textit{dest}(e_k)=t$ \cite{sriy90}.} $\HGAP$: set $\calT=\{u_1\dand u_2\dsub v\mid (u_1,u_2;v)\in E\}$, assume w.l.o.g.\ the set of source nodes as $S=\{s\}$, then $(G,S,t)\in\HGAP$ iff $(\calT,s,t)\in\SUBS_\emptyset(\CloneE_2)$. For the lower bound of $\CloneV_2$ apply \Cref{lem:contrapositionSUBS}.
		\item Follows directly by the reduction from $\IMP(\CloneL_2)$ due to \Cref{thm:imp} and \Cref{lem:subs-csat}(\ref{lemitem:imp_lowerbound}.).
		\item For the lower bound we show a reduction from the graph accessibility problem\footnote{Instances of $\GAP$ are directed graphs $G$ together with two nodes $s,t$ in $G$ asking whether there is a path from $s$ to $t$ in $G$.} $\GAP$ to $\SUBS_\emptyset(\CloneI_2)$.
Let $G=(V,E)$ be a undirected graph and $s,t\in V$ be the vertices for the input. Then for $\calT:=\{(A_u\dsub A_v)\mid(u,v)\in E\}$ it holds that $(G,s,t)\in\GAP$ iff $(\calT,A_s,A_t)\in\SUBS_\emptyset(\CloneI_2)$. 

		For the upper bound we follow the idea from \cite[Lemma 29]{ms11b}. Given the input instance $(\calT,C,D)$ we can similarly assume that for each $E\dsub F\in\calT$ it holds that $E,F$ are atomic concepts, or their negations, or  constants. Now $(\calT,C,D)\in\SUBS_\emptyset(\CloneN)$ holds iff for every interpretation $\calI=(\Delta^\calI,\cdot^\calI)$ and $x\in\Delta^\calI$ it holds that if $x\in C^\calI$ then $x\in D^\calI$ holds iff for the implication graph $G_\calT$ (constructed as in \cite[Lemma 29]{ms11b}) there exists a path from $v_C$ to $v_D$. 
		
		Informally if there is no path from $v_C$ to $v_D$ then $D$ is not implied by $C$, \ie, it is possible to construct an interpretation for which there exists an individual which is a member of $C^\calI$ but not of $D^\calI$.
		
		Thus we have provided a $\co\NL$-algorithm which first checks accordingly to the algorithm in \cite[Lemma 29]{ms11b} if there are not any cycles containing contradictory axioms. Then we verify that there is no path from $v_C$ to $v_D$ implying that $C$ is not subsumed by $D$.
	\end{enumerate}
\end{proof}

Using some results from the previous theorem we are now able to classify most fragments of the subsumption problem using only either the $\forall$ or $\exists$ quantifier with respect to all possible Boolean clones in the following two theorems.

\begin{theorem}[Universal fragments]\label{thm:subs_all} \NoEndMark
Let $B$ be a finite set of Boolean operators.
\begin{enumerate}
	\item If $\CloneE_2\subseteq[B]$, then $\SUBS_\forall(B)$ is $\EXPTIME$-complete.\label{num:subs_all_E2}
	\item If $\CloneN_2\subseteq [B]$ or $\CloneL_0\subseteq[B]$, then $\SUBS_\forall(B)$ is $\EXPTIME$-complete.\label{subs_all_L0_EXP}
	\item If $\CloneL_1\subseteq[B]$, then $\SUBS_\forall(B)$ is $\EXPTIME$-complete.
	\item If $\CloneS_{00}\subseteq[B]$, then $\SUBS_\forall(B)$ is $\EXPTIME$-complete.
	\item If $\CloneD_2\subseteq[B]\subseteq\CloneD_1$, then $\SUBS_\forall(B)$ is $\co\NP$-hard and in $\EXP$.
	\item If $[B]\subseteq\CloneV$, then $\SUBS_\forall(B)$ is $\P$-complete.\label{num:subs_forall_P-complete}
	\item If $[B]=\CloneL_2$, then $\SUBS_\forall(B)$ is $\P$-hard and in $\EXP$.
\end{enumerate}
All hardness results hold \wrt $\leqlogm$ reductions.
\end{theorem}

\begin{proof}
	\begin{enumerate}
		\item Follows from $\EXPTIME$-hardness of $\FLzero\text{-}\SUBS$\index{FLzero@$\FLzero$} which has been shown in \cite[Thm 7.6]{hof05}.
		\item The lower bound for $\CloneN_2\subseteq[B]$ is achieved through the reductions $$\overline{\TCSAT_\forall(\CloneN_2)}\leqlogm\SUBS_\forall(\CloneN)\equivlogm\SUBS_\forall(\CloneN_2),$$ where the first reduction is due to \Cref{lem:subs-csat}(\ref{lemitem:coTCSAT_lowerbound}.) and the second equivalence holds through \Cref{lem:topbot-always-above-negSUBS} which enables us to always have access to both constants whenever $\CloneN_2\subseteq[B]$. The $\EXP$-hardness now follows from $\TCSAT_\forall(\CloneN_2)$ being $\EXP$-complete proven in \cite[Theorem 32 (1.)]{ms11b}.
		
		The $\EXP$-hardness for $\CloneL_0\subseteq[B]$ follows from \Cref{lem:subs-csat}(\ref{lemitem:coTCSAT_lowerbound}.) which states the reduction $\overline{\TCSAT_\forall(\CloneL_0)}\leqlogm\SUBS_\forall(\CloneL_0\cup\{\false\})$ where $[\CloneL_0\cup\{\false\}]=\CloneL_0$. From \cite[Theorem 32 (1.)]{ms11b} we know that $\TCSAT_\forall(\CloneL_0)$ is $\EXP$-complete.
		\item The $\EXP$-hardness follows from the following reduction:
		$$
		\overline{\TCSAT_\exists(\CloneL_0)}\overset{(a)}{\leqlogm}\SUBS_\exists(\CloneL_0)\overset{(b)}{\leqlogm}\SUBS_{\dual\exists}(\dual{\CloneL_0})=\SUBS_\forall(\CloneL_1),
		$$
		by virtue of $\TCSAT_\exists(\CloneL_0)$ being $\EXP$-complete shown in \cite[Theorem 32 (1.)]{ms11b} for $(a)$, and \Cref{lem:contrapositionSUBS} for $(b)$.
		\item Follows from \Cref{lem:constants-for-s00-s01} and the $\EXP$-hardness of $\SUBS_\forall(\CloneM_0)$ overlaid by \Cref{thm:subs_all}(\ref{num:subs_all_E2}.), as $\CloneM_0=[\CloneS_{00}\cup\{\false\}]$.
		\item The $\co\NP$-hardness follows from $\SUBS_\emptyset(\CloneD_2)$ being $\co\NP$-hard shown in  \Cref{thm:subs_empty}.
		\item For the upper bound \Cref{lem:contrapositionSUBS} lets us state the reduction
		$
		\SUBS_\forall(\CloneV)\leqlogm\SUBS_\exists(\CloneE)
		$, where the latter is in $\P$ by virtue of \Cref{thm:subs_ex}(\ref{num:subs_exists_P-complete}.).
		
		The lower bound follows again from \Cref{lem:contrapositionSUBS}, and the $\P$-hardness of $\SUBS_\exists(\CloneI_2)$ which is proven in \cite[Lemma 26]{ms11b}.
		\item The $\P$-hardness follows from(\ref{num:subs_forall_P-complete}.).
	\end{enumerate}
\end{proof}

\begin{theorem}[Existential fragments]\label{thm:subs_ex} \NoEndMark
Let $B$ be a finite set of Boolean operators.
\begin{enumerate}
	\item If $\CloneV_2\subseteq[B]$, then $\SUBS_\exists(B)$ is $\EXPTIME$-complete.\label{num:subs_ex_V2}
	\item If $\CloneN_2\subseteq [B]$ or $\CloneL_0\subseteq[B]$, then $\SUBS_\exists(B)$ is $\EXPTIME$-complete.
	\item If $\CloneL_1\subseteq[B]$, then $\SUBS_\exists(B)$ is $\EXPTIME$-complete.
	\item If $\CloneS_{10}\subseteq[B]$, then $\SUBS_\exists(B)$ is $\EXPTIME$-complete.
	\item If $\CloneD_2\subseteq[B]\subseteq\CloneD_1$, then $\SUBS_\exists(B)$ is $\co\NP$-hard and in $\EXP$.
	\item If $[B]\subseteq\CloneE$, then $\SUBS_\exists(B)$ is $\P$-complete.\label{num:subs_exists_P-complete}
	\item If $[B]=\CloneL_2$, then $\SUBS_\exists(B)$ is $\P$-hard and in $\EXP$.
\end{enumerate}
All hardness results hold \wrt $\leqlogm$ reductions.
\end{theorem}

\begin{proof}
\begin{enumerate}
 \item[1.-3.] For the following reductions showing the needed lower bounds for $\CloneV_2,\CloneN_2,\CloneL_1,$ and $\CloneL_0$ we use \Cref{thm:subs_all} in combination with the contraposition argument in \Cref{lem:contrapositionSUBS}:
		\begin{align*}
			\SUBS_\forall(\CloneE_2)&\leqlogm\SUBS_\exists(\CloneV_2), &
			\SUBS_\forall(\CloneN_2)&\leqlogm\SUBS_\exists(\CloneN_2),\\
			\SUBS_\forall(\CloneL_0)&\leqlogm\SUBS_\exists(\CloneL_1),\text{ and}&
			\SUBS_\forall(\CloneL_1)&\leqlogm\SUBS_\exists(\CloneL_0).
		\end{align*}

\item[4.] The needed lower bound follows from \Cref{lem:constants-for-s00-s01} whereas the $\EXP$-hardness of $\SUBS_\exists(\CloneM_1)$ overlaid by \Cref{thm:subs_ex}(\ref{num:subs_ex_V2}.) as $\CloneM_1=[\CloneS_{10}\cup\{\true\}]$.

\item[5.] The $\co\NP$ lower bound follows from $\SUBS_\emptyset(B)$ shown in \Cref{thm:subs_empty}.

\item[6.] The upper bound follows from the membership of subsumption for the logic $\mathcal{ELH}$\index{ELH@$\mathcal{ELH}$} in \P, \cite[Thm. 9]{bra04c}. The lower bound is proven in \cite[Lemma 26]{ms11b}.

\item[7.] The lower bound follows from \Cref{thm:subs_ex}(\ref{num:subs_exists_P-complete}.).
\end{enumerate}
\end{proof}

Finally the classification of the full quantifier fragments naturally emerges from the previous cases to $\EXP$-complete, $\co\NP$-, and $\P$-hard cases.

\begin{theorem}[Both quantifiers available] \NoEndMark
Let $B$ be a finite set of Boolean operators.
\begin{enumerate}
	\item Let $X\in\{\CloneN_2,\CloneV_2,\CloneE_2\}$. If $X\subseteq[B]$, then $\SUBS_\exall(B)$ is \EXPTIME-complete.
	\item If $\CloneI_0\subseteq[B]$ or $\CloneI_1\subseteq[B]$, then $\SUBS_\exall(B)$ is \EXPTIME-complete.
	\item If $\CloneD_2\subseteq[B]\subseteq\CloneD_1$, then $\SUBS_\exall(B)$ is \co\NP-hard and in \EXP.
	\item If $[B]\in\{\CloneI_2,\CloneL_2\}$, then $\SUBS_\exall(B)$ is \P-hard and in \EXP.
\end{enumerate}
All hardness results hold \wrt $\leqlogm$ reductions.
\end{theorem}

\begin{proof}
\begin{enumerate}
\item Follows from the respective lower bounds of $\SUBS_\exists(B)$, \resp, $\SUBS_\forall(B)$ shown in \Cref{thm:subs_all,thm:subs_ex}.
		
\item The needed lower bound follows from \Cref{lem:subs-csat}(\ref{lemitem:coTCSAT_lowerbound}.) and enables a reduction from the $\EXPTIME$-complete problem $\overline{\TCSAT_\exall(\CloneI_0)}$ \cite[Theorem 2 (1.)]{ms11}. The case $\SUBS_\exall(B)$ with $\CloneI_1\subseteq[B]$ follows from the contraposition argument in \Cref{lem:contrapositionSUBS}.
		
\item[3.+4.] The lower bounds carry over from $\SUBS_\emptyset(B)$ for the respective sets $B$ (see \Cref{thm:subs_empty}).
\end{enumerate}
\end{proof}

\section{Conclusion and Discussion}
\Cref{fig:SUBSexall_ex_all} visualizes how the results arrange in Post's lattice. The classification has shown that the subsumption problem with both quantifiers is a very difficult problem. Even a restriction down to only one of the constants leads to an intractable fragment with $\EXPTIME$-completeness. Although we achieved a $\P$ lower bound for the case without any constants, i.e., the clone $\CloneI_2$ it is not clear how to state a polynomial time algorithm for this case: We believe that the size of satisfying  interpretations always can be polynomially in the size of the given TBox but a deterministic way to construct it is not obvious to us. If one starts with an individual instantiating the given concept $C$ then it is not easy to decide how to translate a triggered axiom into the interpretation (e.g., should a role edge be a loop or not). Further it is much harder to construct such an algorithm for the case $\CloneL_2$ having a ternary exclusive-or operator.

Retrospectively the subsumption problem is much harder than the the usual terminology satisfiability problems visited in \cite{ms11}. Due to the duality principle expressed by \Cref{lem:contrapositionSUBS} both halves of Post's lattice contain intractable fragments plus it is not clear if there is a tractable fragment at all. For the fragments having access to only one of the quantifiers the clones which are able to express either disjunction (for the universal quantifier) or conjunction (for the existential case) become tractable (plus both constants). Without any quantifier allowed the problem almost behaves as the propositional implication problem with respect to tractability. The only exception of this rule are the $\CloneL$-cases that can express negation or at least one constant. They become $\co\NP$-complete and therewith intractable.

Finally a similar systematic study of the subsumption problem for concepts (without respect to a terminology) would be of great interest because of the close relation to the implication problem of modal formulae. To the best of the author's knowledge such a study has not been done yet and would enrich the overall picture of the complexity situation in this area of research. Furthermore it would be interesting to study the effects of several restrictions on terminologies to our classification, e.g., acyclic or cyclic TBoxes.
\begin{figure}
\centering
\includegraphics[width=.9\linewidth]{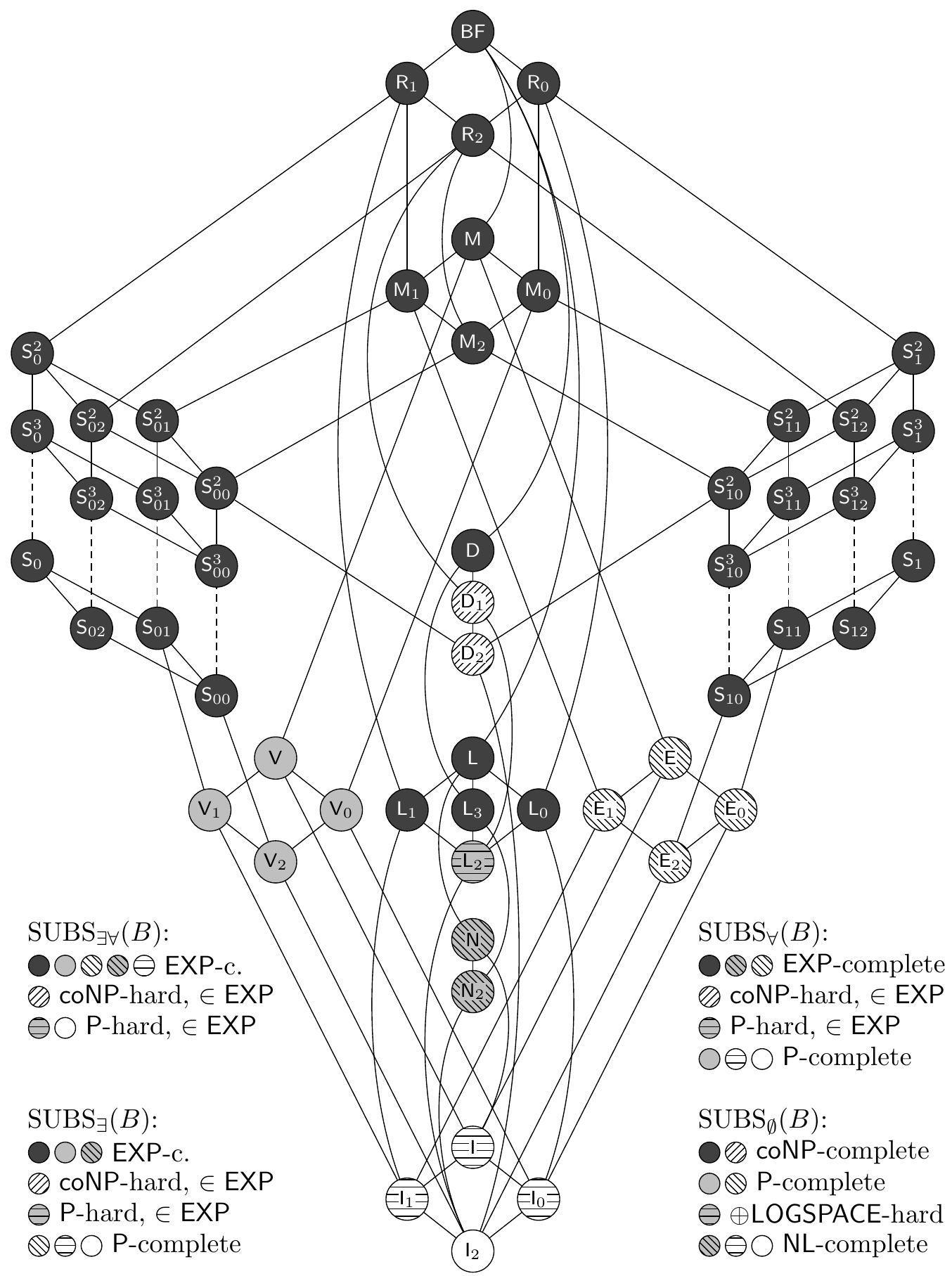}
	\caption{Post's lattice showing the complexity of $\SUBS_\calQ(B)$ for all sets $\emptyset\subseteq\calQ\subseteq\{\exists,\forall\}$ and all Boolean clones $[B]$.}
	\label{fig:SUBSexall_ex_all}
\end{figure}

\paragraph{Acknowledgements}
The author thanks Thomas Schneider (Bremen) and Peter Lohmann (Hannover) for several helpful discussions about the paper.

\bibliographystyle{plain}
\bibliography{subsumption}

\begin{thebibliography}{10}

\bibitem{ackz07}
A.~Artale, D.~Calvanese, R.~Kontchakov, and M.~Zakharyaschev.
\newblock {DL-Lite} in the light of first-order logic.
\newblock In {\em Proc.\ AAAI}, pages 361--366, 2007.

\bibitem{ackz09}
A.~Artale, D.~Calvanese, R.~Kontchakov, and M.~Zakharyaschev.
\newblock Adding weight to {DL-Lite}.
\newblock In {\em Proc.\ DL}, CEUR-WS, 2009.

\bibitem{bbl05}
F.~Baader, S.~Brandt, and C.~Lutz.
\newblock Pushing the {\EL} envelope.
\newblock In {\em Proc.\ IJCAI}, pages 364--369, 2005.

\bibitem{bbl08}
F.~Baader, S.~Brandt, and C.~Lutz.
\newblock Pushing the {\EL} envelope further.
\newblock In {\em Proc.\ OWLED DC}, 2008.

\bibitem{DLHB}
F.~Baader, D.~Calvanese, D.~L. McGuinness, D.~Nardi, and P.~F. Patel-Schneider,
  editors.
\newblock {\em The Description Logic Handbook: Theory, Implementation, and
  Applications}, volume~1.
\newblock Cambridge University Press, 2nd edition, 2003.

\bibitem{bemethvo09}
O.~Beyersdorff, A.~Meier, M.~Thomas, and H.~Vollmer.
\newblock The complexity of reasoning for fragments of default logic.
\newblock In {\em Theory and Applications of Satisfiability Testing - SAT
  2009}, volume 5584 of {\em Lecture Notes in Computer Science}, pages 51--64.
  Springer Berlin / Heidelberg, 2009.

\bibitem{bmtv09}
O.~Beyersdorff, A.~Meier, M.~Thomas, and H.~Vollmer.
\newblock {The Complexity of Propositional Implication}.
\newblock {\em Information Processing Letters}, 109(18):1071--1077, 2009.

\bibitem{Bohler:2003tg}
E~B{\"o}hler, N~Creignou, S~Reith, and H~Vollmer.
\newblock {Playing with Boolean blocks, part I: Post's lattice with
  applications to complexity theory}.
\newblock {\em SIGACT News}, 34(4):38--52, 2003.

\bibitem{brle84}
R.~J. Brachman and H.~J. Levesque.
\newblock The tractability of subsumption in frame-based description languages.
\newblock In {\em AAAI}, pages 34--37, 1984.

\bibitem{brsc85}
R.~J. Brachman and J.~G. Schmolze.
\newblock An overview of the {\sc kl-one} knowledge representation system.
\newblock {\em Cognitive Science}, 9(2):171--216, 1985.

\bibitem{bra04c}
S.~Brandt.
\newblock Subsumption and instance problem in $\mathcal{ELH}$ w.r.t. general
  tboxes.
\newblock LTCS-Report LTCS-04-04, Chair for Automata Theory, Institute for
  Theoretical Computer Science, Dresden University of Technology, Germany,
  2004.
\newblock See http://lat.inf.tu-dresden.de/research/reports.html.

\bibitem{cgllr05}
D.~Calvanese, G.~{De Giacomo}, D.~Lembo, M.~Lenzerini, and R.~Rosati.
\newblock {DL-Lite}: Tractable description logics for ontologies.
\newblock In {\em Proc.\ AAAI}, pages 602--607, 2005.

\bibitem{19007439}
Ronald Cornet and Nicolette de~Keizer.
\newblock Forty years of snomed: a literature review.
\newblock {\em BMC Medical Informatics and Decision Making}, 8(Suppl 1):S2,
  2008.

\bibitem{cmtv10}
N.~Creignou, A.~Meier, M.~Thomas, and H.~Vollmer.
\newblock The complexity of reasoning for fragments of autoepistemic logic.
\newblock In Benjamin Rossman, Thomas Schwentick, Denis Th{\'e}rien, and
  Heribert Vollmer, editors, {\em Circuits, Logic, and Games}, number 10061 in
  Dagstuhl Seminar Proceedings, Dagstuhl, Germany, 2010. Schloss Dagstuhl -
  Leibniz-Zentrum fuer Informatik, Germany.

\bibitem{crscthwo10}
N.~Creignou, J.~Schmidt, M.~Thomas, and S.~Woltran.
\newblock Sets of boolean connectives that make argumentation easier.
\newblock In {\em Proc. 12th European Conference on Logics in Artificial
  Intelligence}, volume 6341 of {\em Lecture Notes in Computer Science}, pages
  117--129. Springer, 2010.

\bibitem{dlnhnm92}
F.~M. Donini, M.~Lenzerini, D.~Nardi, B.~Hollunder, W.~Nutt, and
  A.~Marchetti-Spaccamela.
\newblock The complexity of existential quantification in concept languages.
\newblock {\em AI}, 53(2-3):309--327, 1992.

\bibitem{dlnn97}
F.~M. Donini, M.~Lenzerini, D.~Nardi, and W.~Nutt.
\newblock The complexity of concept languages.
\newblock {\em Inf. Comput.}, 134(1):1--58, 1997.

\bibitem{hescsc08}
E.~Hemaspaandra, H.~Schnoor, and I.~Schnoor.
\newblock Generalized modal satisfiability.
\newblock {\em CoRR}, abs/0804.2729:1--32, 2008.

\bibitem{hof05}
M.~Hofmann.
\newblock Proof-theoretic approach to description-logic.
\newblock In {\em Proc.\ LICS}, pages 229--237, 2005.

\bibitem{MS10}
A.~Meier and T.~Schneider.
\newblock The complexity of satisfiability for sub-{B}oolean fragments of
  {ALC}.
\newblock In {\em Proc.\ of DL-2010}. CEUR-WS.org, 2010.

\bibitem{ms11}
A.~Meier and T.~Schneider.
\newblock Generalized satisfiability for the description logic $\mathcal{ALC}$.
\newblock In {\em Proceedings of the 8th Annual Conference on Theory and
  Application of Models of Computation}, volume 6648 of {\em Lecture Notes in
  Computer Science}, pages 552--562. Springer Verlag, 2011.

\bibitem{ms11b}
A.~Meier and T.~Schneider.
\newblock Generalized satisfiability for the description logic $\mathcal{ALC}$.
\newblock {\em CoRR}, abs/1103.0853:1--37, March 2011.

\bibitem{owl2}
B.~Motik, P.~F. Patel-Schneider, and B.~Parsia.
\newblock Owl 2 web ontology language: Structural specification and
  functional-style syntax, 2009.
\newblock \url{http://www.w3.org/TR/2009/REC-owl2-syntax-20091027/}.

\bibitem{nabra03}
D.~Nardi and R.~J. Brachman.
\newblock {\em An Introduction to Description Logics}, chapter~1.
\newblock Volume~1 of Baader et~al. \cite{DLHB}, 2nd edition, 2003.

\bibitem{pip97b}
N.~Pippenger.
\newblock {\em Theories of Computability}.
\newblock Cambridge University Press, 1997.

\bibitem{pos41}
E.~Post.
\newblock The two-valued iterative systems of mathematical logic.
\newblock {\em Annals of Mathematical Studies}, 5:1--122, 1941.

\bibitem{19007443}
Patrick Ruch, Julien Gobeill, Christian Lovis, and Antoine Geissbuhler.
\newblock Automatic medical encoding with snomed categories.
\newblock {\em BMC Medical Informatics and Decision Making}, 8(Suppl 1):S6,
  2008.

\bibitem{sc07}
H.~Schnoor.
\newblock {\em Algebraic Techniques for Satisfiability Problems}.
\newblock PhD thesis, Gottfried Wilhelm Leibniz Universit{\"a}t Hannover, 2007.

\bibitem{sc08}
I.~Schnoor.
\newblock {\em The Weak Base Method for Constraint Satisfaction}.
\newblock PhD thesis, Gottfried Wilhelm Leibniz Universit{\"a}t Hannover, 2008.

\bibitem{sriy90}
R.~Sridhar and S.~Iyengar.
\newblock Efficient parallel algorithms for functional dependency
  manipulations.
\newblock In {\em Proc.\ ICPADS}, pages 126--137. ACM, 1990.

\bibitem{vol99}
H.~Vollmer.
\newblock {\em Introduction to Circuit Complexity}.
\newblock Springer, 1999.

\end{thebibliography}

\end{document}